\newif\ifabstract
\abstracttrue
 \abstractfalse 
\newif\iffull
\ifabstract \fullfalse \else \fulltrue \fi

\documentclass[11pt]{article}
\usepackage{xfrac}
\usepackage{tikz}
\usetikzlibrary{tikzmark}

\usepackage{amsfonts}
\usepackage{amssymb}
\usepackage{amstext}
\usepackage{amsmath}
\usepackage{xspace}
\usepackage{theorem}
\usepackage{graphicx}
\usepackage{url}
\usepackage{graphics}
\usepackage{colordvi}
\usepackage{colordvi}
\usepackage{subfigure}
\usepackage[noend]{algpseudocode}

\usepackage{graphicx,amssymb,amsmath}
\usepackage[ruled,vlined,linesnumbered]{algorithm2e}
\usepackage{xcolor}

\advance \oddsidemargin by -0.8in
\newcommand{\myparskip}{3pt}
\parskip \myparskip

        {\hspace*{\fill}$\Box$\par\vspace{4mm}}

\newcommand{\be}{\begin{enumerate}}
\newcommand{\ee}{\end{enumerate}}
\newcommand{\bd}{\begin{description}}
\newcommand{\ed}{\end{description}}
\newcommand{\bi}{\begin{itemize}}
\newcommand{\ei}{\end{itemize}}

\newtheorem{theorem}{Theorem}[section]
\newtheorem{lemma}[theorem]{Lemma}
\newtheorem{Observation}[theorem]{Observation}

\newtheorem{claim}[theorem]{Claim}

\newenvironment{proof}{\par \smallskip{\bf Proof:}}{\hfill\stopproof}
\def\stopproof{\square}
\def\square{\vbox{\hrule height.2pt\hbox{\vrule width.2pt height5pt \kern5pt
\vrule width.2pt} \hrule height.2pt}}

\usepackage{graphicx,amssymb,amsmath}
\usepackage[ruled,vlined]{algorithm2e}
\usepackage{xcolor}

\renewcommand{\phi}{\varphi}

\mathchardef\hyphen="2D

\begin{document}

\title{Max-Min $k$-Dispersion on a Convex Polygon}
\author{Vishwanath R. Singireddy
\and   Manjanna Basappa\thanks{Corresponding author
}}

\date{%
Department of Computer Science \& Information Systems, \\BITS Pilani, Hyderabad Campus, Telangana 500078, India\\{\tt \{p20190420,manjanna\}@hyderabad.bits-pilani.ac.in}\\%
    \today
}


\begin{titlepage}
\maketitle

\thispagestyle{empty}

\begin{abstract}
In this paper we consider the following $k$-dispersion problem. Given a set $S$ of $n$ points placed in the plane in a convex position, and an integer $k$ ($0<k<n$), the objective is to compute a subset $S'\subset S$ such that $|S'|=k$ and the minimum distance between a pair of points in $S'$ is maximized. Based on the bounded search tree method we propose an exact fixed-parameter algorithm in $O(2^k(n^2\log n+n(\log^2 n)(\log k)))$ time, for this problem, where $k$ is the parameter. The proposed exact algorithm is better than the current best exact exponential algorithm [$n^{O(\sqrt{k})}$-time algorithm by Akagi et al.,(2018)] whenever $k<c\log^2{n}$ for some constant $c$. We then present an $O(\log{n})$-time $\frac{1}{2\sqrt{2}}$-approximation algorithm for the problem when $k=3$ if the points are given in convex position order.

\end{abstract}

\end{titlepage}

\section{Introduction}

In many variants of the facility location problems that are studied in the literature \cite{DREZ1995,DREZ2004}, typically we are given a set of $n$ points and among them we need to locate $k$ facilities such that some objective function is minimized. On contrast, in the obnoxious facility location problems, we need to maximize an objective function. In the literature, this wider class of facility location problems wherein the objective is to maximize some diversity measure are called dispersion problems. In the case of the max-min $k$-dispersion problem, we need to maximize the minimum distance between the selected $k$ facilities.

The applications of $k$-dispersion problems arise in many areas. Consider a specific application where the $k$-dispersion problem can be used in which the given points are in a convex position, as discussed below. Consider an island where some nuclear plants or oil storage plants are to be established, and these plants should be kept as far away from each other as possible so that any accident in one plant should not affect the other plant. Hence, we can model this problem as the $k$-dispersion problem, wherein the above plants are placed on the boundary of the island to maximize the distance between any pair of the plants.  

Another application is in military sector for making strategic plans. Suppose we have $k$ military squads and a convex region. The problem is to find locations for placing these squads on the boundary of the convex region. Now, we can model this as the $k$-dispersion problem on a convex polygon to maximize the distance between any two squads. Hence, an attack on one squad will not affect the other squads, and nearby squads will get time to respond to the attack.  

\section{Literature survey}

The discrete $k$-dispersion problem for $k\geq3$ is known to be {\tt NP-complete} even when the triangle inequality is satisfied \cite{ERKU1990}. The Euclidean $k$-dispersion problem is proved {\tt NP-hard} by Wang and Kuo \cite{WK1988}. Akagi et al. \cite{AKAG2018} gave an algorithm to solve the $k$-dispersion problem in the Euclidean plane exactly in $n^{O(\sqrt{k})}$ time. They also gave an $O(n)$-time algorithm to solve the special cases of the problem in which the given points appear in order on a line or on the boundary of a circle. Later, Araki and Nakano \cite{ARAK2018} improved the running time of \cite{AKAG2018} to $O(\log{n})$ for the line case. Ravi et al. \cite{RAVI1994} proved that for the max-min $k$-dispersion problem on an arbitrary weighted graph, we cannot give any constant factor approximation algorithm within polynomial time unless {\tt P}={\tt NP}. If the triangle inequality is satisfied by the edge weights, then we cannot approximate the problem with a better factor than $\frac{1}{2}$ in polynomial time unless {\tt P}={\tt NP}. They also gave a polynomial time $\frac{1}{2}$-approximation algorithm for the problem in graph metric.

Horiyama et al. \cite{HORI2019} solved the max-min 3-dispersion problem in $O(n)$ time in both $L_1$ and $L_{\infty}$ metrics when the given points are in  2-dimensional plane. They also designed an $O(n^2\log{n})$ time algorithm for the 3-dispersion problem in $L_2$ metric. The 1-dispersion problem is trivial when the points are in a convex position, and we can solve the $2$-dispersion problem in $O(n\log{n})$ time by computing the diameter of the convex polygon formed by these points \cite{SMI78}. Recently, Kobayashi et al. \cite{KOBA2021} gave $O(n^2)$-time algorithm for the  3-dispersion problem on a convex polygon. In the literature, to the best of our knowledge, the $k$-dispersion problem on a convex polygon for any $k>3$ is still open. When the points are arbitrarily placed in the Euclidean plane, the current best approximation algorithm is still the $\frac{1}{2}$-approximation algorithm proposed by Ravi et al. \cite{RAVI1994} for the metric case. Hence, from the point of designing $\rho$-approximation algorithm for $\rho>\frac{1}{2}$ also, the problem is open.

\subsection{Description of the problem and our results}
\noindent {\it Discrete $k$-dispersion on a Convex Polygon ({\textsc{DkConP}}):}
Given a set $S$ of $n$ points in a convex position, assume that the points in $S$ are ordered in a clockwise order around the centroid of $S$. Then, observe that the $k$-dispersion problem on the set $S$ can be equally stated as packing $k$ congruent disks of the maximum radius, with their centers lying at the vertices of the compact convex hull ${\cal P}$ of $S$.

\noindent {\it Our results:}
We propose an exact algorithm for the {\textsc{DkConP}} problem. The running time of the algorithm is $O(2^kn(n+(\log{n})(\log{k}))\log{n})$, which is an improvement over the previous best exact algorithm (running in $n^{O(\sqrt{k})}$) \cite{AKAG2018}. Here, the constant hidden in $O(\sqrt{k})$, the exponent of $n$ in the running time, is larger than $5.44$ \cite{MARK2015}. Consider comparing the running times of the former algorithm with the latter, it turns out that $O(2^kn(n+(\log{n})(\log{k}))\log{n})<n^{O(\sqrt{k})}$ for all $k<c\log^2{n}$, where $c>5.44$ is a constant. Hence, for the range of $k$ for which $k<c\log^2{n}$, the running time of our algorithm is significantly faster. We also give a logarithmic time $0.36$-factor approximation algorithm for $k=3$. The previous best approximation algorithm runs in $\Omega(n^2)$ and computes $\frac{1}{2}$-approximation result \cite{RAVI1994}. For $k=3$, the previous best exact algorithm runs in $O(n^2)$ time \cite{KOBA2021}. Hence, our algorithm computes an approximation result much quicker than the previous best algorithms, by sacrificing a little bit of quality.

\section{Preliminaries}


This section discusses some concepts that are useful in the discussion of our solution for the {\textsc{DkConP}} problem.

Let us use $\overline{v_iv_j}$ to denote the line segment connecting $v_i$ and $v_j$. We use $|$.$|$ (i) to denote the length $|v_iv_j|$ of the line segment $\overline{v_iv_j}$, (ii) to denote the absolute value $|x|$ of a real number $x\in \mathbb{R}$, and also, (iii) to denote the cardinality $|S|$ of any set $S$. The center of any disk $d$ is denoted by ${\cal C}(d)$ and the diameter of ${\cal P}$ is denoted by ${\cal D(\cal P)}$. 

Let ${\cal P}$ be any convex polygon with $n$ vertices and $\Delta$ be a convex polygon with $m$ vertices, inscribed inside ${\cal P}$ such that the vertices of $\Delta$ coincide with some vertices of ${\cal P}$, where $2\leq m< n$. Let $\delta(\Delta)$ denote the length of the smallest edge or diagonal of $\Delta$, i.e., $\delta(\Delta)=\min\limits_{i,j, i\neq j}\{|v_iv_j|\}$. Now, we can easily observe that to solve the {\textsc{DkConP}} problem optimally on ${\cal P}$, given $\Delta$ inscribed in ${\cal P}$ we need to maximize the $\delta(\Delta)$ by changing the positions of the vertices of $\Delta$ to coincide with other vertices of ${\cal P}$. Then, we can pack $k$ disks with centers at the vertices of this modified ${\Delta}$ since these disks do not overlap if the radius of the disks is $\frac{\delta(\Delta)}{2}$. Hence, we state the following obvious observations.

\begin{Observation}\label{observation-1}
Given a convex polygon ${\cal P}$ with $n$ vertices, an integer $k$, and a convex polygon $\Delta$ with $k$ $(k\leq n)$ vertices inscribed inside ${\cal P}$ such that the vertices of $\Delta$ are at the vertices of ${\cal P}$, then the maximum  radius $r_{max}$ of $k$ congruent disks packed with centers lying at the vertices of ${\cal P}$ is $\frac{\delta(\Delta)}{2}$ if and only if $\delta(\Delta)$ is maximized.
\end{Observation} 

It is easy to see that any convex polygon $\Delta$ with $k$ vertices $v_{i_1}$, $v_{i_2}$, $\ldots$, $v_{i_k}$, inscribed inside 
 ${\cal P}$ such that $v_{i_1}$, $v_{i_2}$, $\ldots$, $v_{i_k}$ coincide with some vertices of ${\cal P}$, will have edges 
 which are either edges of ${\cal P}$ or chords/diagonals of ${\cal P}$. Let $\overline{v_iv_j}$ be a chord of ${\cal P}$  corresponding to the pair $(v_i, v_j)$ of vertices of ${\cal P}$, where $1<|i-j|<n-1$ and $i, j\in \{1, 2, \ldots, n\}$.
Let $C=\{\overline{v_iv_j}$: $1<|i-j|<n-1$, $i, j\in \{1, 2, \ldots, n\}\}\cup \{\overline{v_1v_2}, \overline{v_2v_3}, \ldots, \overline{v_nv_1}\}$ be the 
set of chords and edges of ${\cal P}$, where 
$\overline{v_iv_{i+1}}$, for $i=1, 2, \ldots, n$, are the edges of ${\cal P}$ and $\overline{v_nv_{n+1}}=\overline{v_nv_1}$. Clearly, $|C|=n(n-2)$. Now, let $C'=\{|v_iv_j| \mid i,j=1,2,\dots,n$ and $\overline{v_iv_j}\in C \}$ be the set of all distinct distances.

\begin{Observation} \label{observation-3}
$2r_{max}\in C'$ and $|C|=|C'|=O(n^2)$.
\end{Observation}

 Due to Observation \ref{observation-3}, we can find $r_{max}$ in at most $\lceil 2{\log{n}} \rceil$ stages of the binary search, provided that for any given $r$ we can decide whether $r>r_{max}$ or $r<r_{max}$. In the following, based on bounded search tree we propose a fixed-parameter algorithm to answer this decision question, where $k$ is the parameter.

\section{An exact fixed-parameter algorithm}

For the {\textsc{DkConP}} problem, our aim here is to develop a fixed-parameter algorithm using the bounded search tree method. In order to do this, we first consider the following decision problem:
 
 \begin{enumerate}
  
 \item[$ $]\textsc{Decision}$({\cal P},k,r)$: Given a convex polygon ${\cal P}$ with $n$ vertices and a positive integer $k < n$ and a radius $r$, is it possible to pack $k$ (non-overlapping) congruent disks of radius $r$, with centers lying at the vertices of ${\cal P}$?
  \end{enumerate}

  \noindent Observe that the answer to \textsc{Decision}$({\cal P},k,r)$ is \textsc{yes} if the radius $r$ is less than or equal to the radius $r_{max}$ of the disks in an optimal solution of the {\textsc{DkConP}} problem. Now, we shall design an algorithm that solves \textsc{Decision}$({\cal P},k,r)$ in $O(f(k)\cdot n^{O(1)})$ time and returns a set of $k$ disks of radius $r$ packed on the boundary of ${\cal P}$ if the answer is \textsc{yes}, and returns \textsc{no} otherwise, where $f(k)$ is an arbitrary exponential function in $k$.

\subsection{Decision algorithm:} 
   The outline of the algorithm is as follows. First, we align the polygon ${\cal P}$ such that its left most vertex $v_1$ is placed at the origin. Then, we place the disk $d_1$ of radius $r$ centered at $v_1$ (see Fig. \ref{figure-89}(a), \ref{figure-89}(b)). From the vertex $v_1$ in clockwise direction along the boundary of ${\cal P}$ we find the first vertex $u$ at which we can center a $r$-radius disk $d_2$ that does not overlap with any previously placed disks (see Fig. \ref{figure-89}(c)). Now, we again have two ways to place the next disk $d_3$, namely, moving in clockwise direction from the center of $d_2$ and moving in counter-clockwise direction from the center of $d_1$. Similarly, from the vertex $v_1$ in counter-clockwise direction along the boundary of ${\cal P}$ we find the first vertex $u'$ at which we could center the disk $d_2$ (see Fig. \ref{figure-89}(d)). In this way, our search for finding all $k-1$ vertices of ${\cal P}$ (as the centers) to pack the disks proceeds like a 2-way search tree. The depth of the search tree is $k$ because we stop after placing $k$ disks and return the disks. At any point along a path of the 2-way search tree if we can not place a disk, then we backtrack to placing a disk in the other direction. Thus, the branching factor of every node is at most 2, resulting in $O(2^k)$ nodes in total. We repeat the above procedure by placing the disk $d_1$ at each of the $n$ vertices of ${\cal P}$. Note that the disks corresponding to the vertices of any path of length $\geq k$ in the 2-way search tree together form a feasible solution for the {\textsc{DkConP}} problem.
  
  \begin{figure}[!htb]
\centering
\includegraphics[scale=0.45]{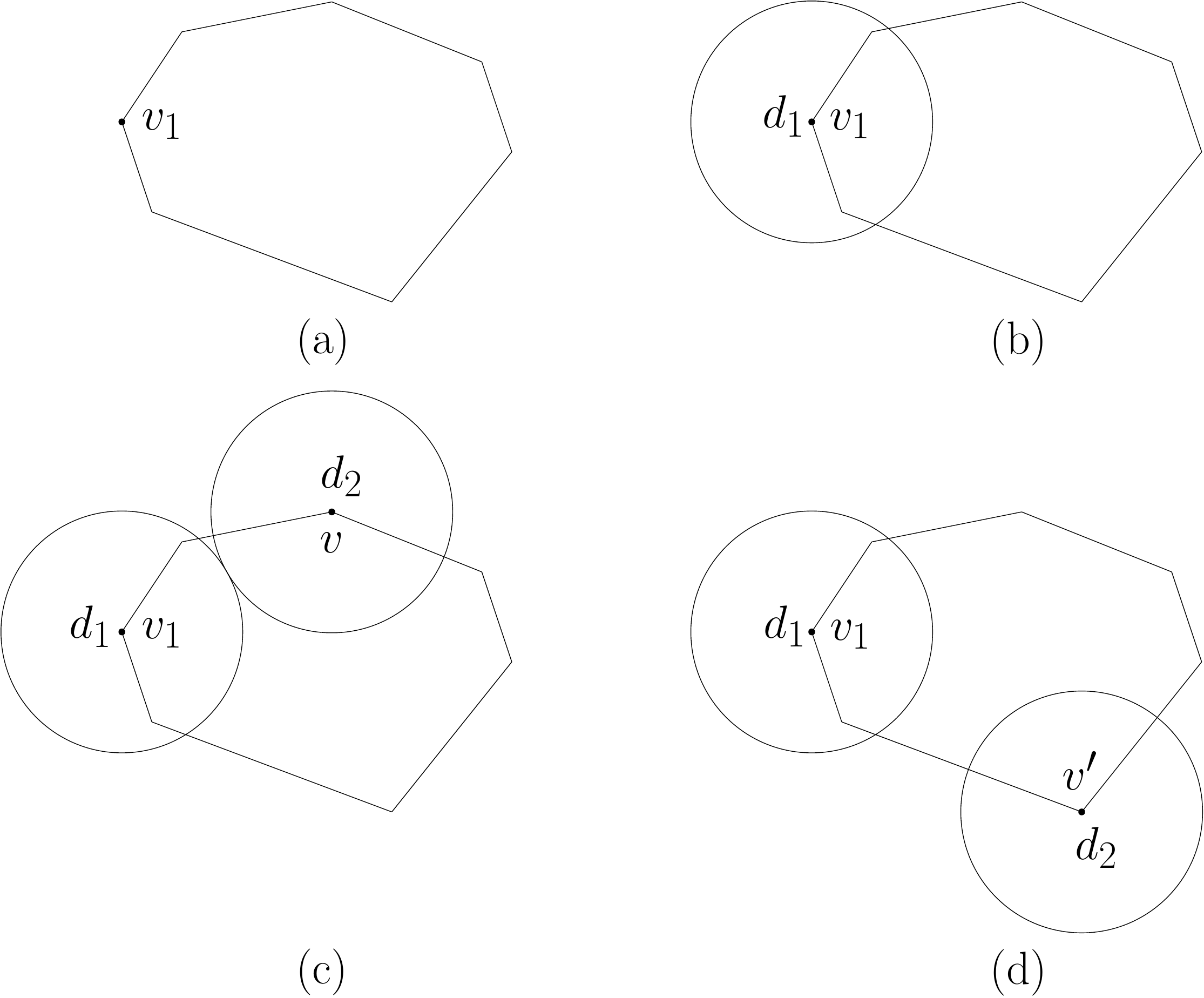}
\caption{Algorithm for \textsc{Decision}$({\cal P},k,r)$ in action}
\label{figure-89}
\end{figure}   

Now, we shall describe how to pack the next disk $d_{j+1}$ after having packed the disks $d_1$, $d_2$, \ldots, $d_j$. In the above 2-way search tree procedure for $j=1, 2, \ldots, k-1$, after packing the disk $d_j$ centered at some vertex of ${\cal P}$, the candidate vertices $u$ and $u'$ for the center of the disk $d_{j+1}$ can be computed as follows: let $u$ be the first vertex at a distance of at least $2r$ from the center of the most recently packed disk ($d_j$ or $d_{j-1}$) clockwise from the center of $d_1$. Similarly, let $u'$ be the first vertex at distance at least $2r$ from the center of the most recently packed disk ($d_j$ or $d_{j-1}$) counter clockwise from the center of $d_1$. Note that $u$ and $u'$ are the two candidate vertices for packing the next disk $d_{j+1}$ which will be centered at one of them. However, it is required to ensure that the distance between the candidate center vertex $u$ or $u'$ and the vertices at which the already-packed disks $d_1, d_2, \ldots, d_j$ are centered is at least $2r$. Observe that for a convex polygon $\cal P$, the distances between a fixed vertex and the remaining vertices of $\cal P$ form a multi-modal function. Therefore, we can not directly employ binary search to find the center vertices $u$ and $u'$ for packing the next disk $d_{j+1}$, $j=1, 2, \ldots, k-1$. For a vertex $v_i$ of $\cal P$ there are $\gamma$ vertices which are the modes or local maxima \cite{ADTGTBB1982}, where $\gamma\leq n/2$. We will exploit this property to identify all the candidate center vertices and to quickly locate the one among them for centering the disk $d_{j+1}$. Hence, given a candidate radius $r$ we do some preprocessing before we shall call the decision algorithm.

\subsection{The optimization scheme:}
To solve the optimization problem, i.e., to find the maximum value $r_{max}$ of $r$, we solve \textsc{Decision}$({\cal P},k,r)$ repeatedly while performing binary search on $C'$. In each stage of the binary search, the radius $r$ will be the median element of $C'$ divide by 2. The median will be found by using a linear time median finding algorithm \cite{CORM2009}. We then perform the above 2-way search tree based procedure, to find an answer to \textsc{Decision}$({\cal P},k,r)$. If the answer is \textsc{yes}, then we update $C'$ by removing all the elements of it that are smaller than $2r$. Otherwise we update by removing all the elements that are at least $2r$. In either case the size of the updated $C'$ will be half of the previous $C'$. The main routine of the algorithm is outlined in Algorithm \ref{algorithm1}.

 \subsubsection*{Preprocessing:} Here we describe how to precompute all the candidate center vertices for the next disk $d_{j+1}$ ($j=1, 2, 3, \ldots, k-1$) once an element $2r\in C'$ is fixed. We also see how to use this precomputed information in every $(j+1)$th step of the decision algorithm after the disks $\{ d_1, d_2, \ldots, d_j\}$ are packed on $\partial {\cal P}$ with centers ${\cal C}(d_1)=v_{\alpha_1}, {\cal C}(d_2)=v_{\alpha_2}, \ldots, {\cal C}(d_j)=v_{\alpha_j}$, where $j=1, 2, \ldots, k-1$.
 
 In Algorithm \ref{algorithm1}, we initially set $\mathcal{X}_1= C'$, the set of all distances between the points in $S$. In the $i$th stage of the binary search (while loop in Algorithm \ref{algorithm1}), we have that $|\mathcal{X}_i|\leq \frac{|\mathcal{X}_{i-1}|}{2}$. Now, consider a straight line $\ell_s$ through $v_s$ of ${\cal P}$, that splits ${\cal P}$ into two parts each with at least one vertex other than $v_s$. Given a median $2r\in \mathcal{X}_i$, for each vertex $v_s$ of ${\cal P}$ the candidate center vertices $u_{s_1}, u_{s_2}, \ldots, u_{s_{\gamma}}$ lying above any straight line $\ell_s$ through $v_s$ (and $u_{s_1'}, u_{s_2'}, \ldots, u_{s_{\gamma'}'}$ lying below $\ell_s$) are such that for $1\leq \beta\leq \gamma$ we have that $|v_{s}u_{s_{\beta}-1}|<2r$, $|v_{s}u_{s_{\beta}+1}|>2r$ and $|v_{s}u_{s_{\beta}}|\geq 2r$ or $|v_{s}u_{s_{\beta}-1}|>2r$, $|v_{s}u_{s_{\beta}+1}|<2r$ and $|v_{s}u_{s_{\beta}}|\geq 2r$, where $\max(\gamma,\gamma')\leq \frac{n}{2}$. These candidate vertices can be pre-computed by doing the distance checks while linearly scanning through $\partial {\cal P}$ both in clockwise and counter clockwise from every vertex $v_i$. Hence, this pre-computation at the beginning of each stage (line 4) will take $O(n^2)$ time. 
The overall time across all stages of the binary search for these pre-computations will be $O(n^2\log n)$.
These candidate center vertices are stored in the array $A_i$ for each vertex $v_i$ (see line 4 of Algorithm \ref{algorithm1}). Let $v_{\alpha_{up}}$ be the vertex in clockwise order from $v_{\alpha_1}$  ($={\cal C}(d_1)$), at which the recently packed disk is centered. Let $v_{\alpha_{low}}$ be the vertex in counter clockwise from $v_{\alpha_1}$, at which the recently packed disk is centered. Let $u_{i_1}, u_{i_2}, \ldots, u_{i_{\gamma}}$ be the candidate center vertices in clockwise order from $v_{\alpha_{up}}$ for packing the next disk $d_{j+1}$. Similarly, the vertices $u_{i_1'}, u_{i_2'}, \ldots, u_{i_{\gamma'}'}$ are the candidate center vertices in counter clockwise order from $v_{\alpha_{low}}$.

\subsubsection*{Computation of a center vertex for $d_{j+1}$ by the decision algorithm:}
Now consider the $(j+1)$th iteration in the $i$th stage of the binary search. Let us denote the right most disks in the packing $\{ d_1, d_2, \ldots, d_{j}\}$ on both upper and lower boundaries of $\partial {\cal P}$ by $d_{\alpha_{up}}$ and $d_{\alpha_{low}}$ centered respectively at $v_{\alpha_{up}}$ and $v_{\alpha_{low}}$. The candidate center vertices from the centers of $d_{\alpha_{up}}$ in clockwise order are $u_{i_1}, u_{i_2}, \ldots, u_{i_{\gamma}}$ and from the center of $d_{\alpha_{low}}$ in counter clockwise order are $u_{i_1'}, u_{i_2'}, \ldots, u_{i_{\gamma'}'}$. Merge these two lists into one single list in convex position order by discarding the candidate centers lying to the left of the line $\ell_{low,up}$ through $v_{\alpha_{low}}$ and $v_{\alpha_{up}}$ (see Fig. \ref{duplow3}). This merging will take $\gamma+\gamma'-1=O(n)$ time. Observe that due to the convexity of ${\cal P}$ each of the vertices ${\cal C}(d_1), {\cal C}(d_2), \ldots, {\cal C}(d_j)$ lie either on $\ell_{low,up}$ or to the left of $\ell_{low,up}$. Assume that the vertices between $u_{i_1}$ and $u_{i_{\gamma'-1}'}$ all have distances at least $2r$ from both $v_{\alpha_{low}}$ and $v_{\alpha_{up}}$, and that $u_{i_{\gamma'}'}$ appears before $u_{i_1}$ in clockwise order from $v_{\alpha_{up}}$. For each $p=i_1, i_1+1, \ldots, i_{\gamma'-1}'$ consider the line $\ell_{p,up}$ through $v_{\alpha_{up}}$ and $v_p$. Since the distances from the line $\ell_{p,up}$ to the vertices ${\cal C}(d_1), {\cal C}(d_2), \ldots, {\cal C}(d_j)$ satisfy unimodality, we can use binary search to discard the vertex $v_p$ if it is of distance strictly less than $2r$ from one of the centers of the disks already packed. Also we use binary search on each subsequence of vertices $u_{i_{\beta}}, u_{i_{\beta}+1}, \ldots, u_{i_{\beta'}'}$ to find the first vertex $v_p$ clockwise from $v_{\alpha_{up}}$ such that ${\cal C}(d_{j+1})=v_p$. This process takes $O(n+(\log n)(\log{j}))$ amortized time. Similarly, we spend the same time if we are packing $d_{j+1}$ in the counter clockwise direction from $v_{\alpha_{low}}$. Then we have the following claim.

\begin{figure}[!htb]
\centering
\includegraphics[scale=0.42]{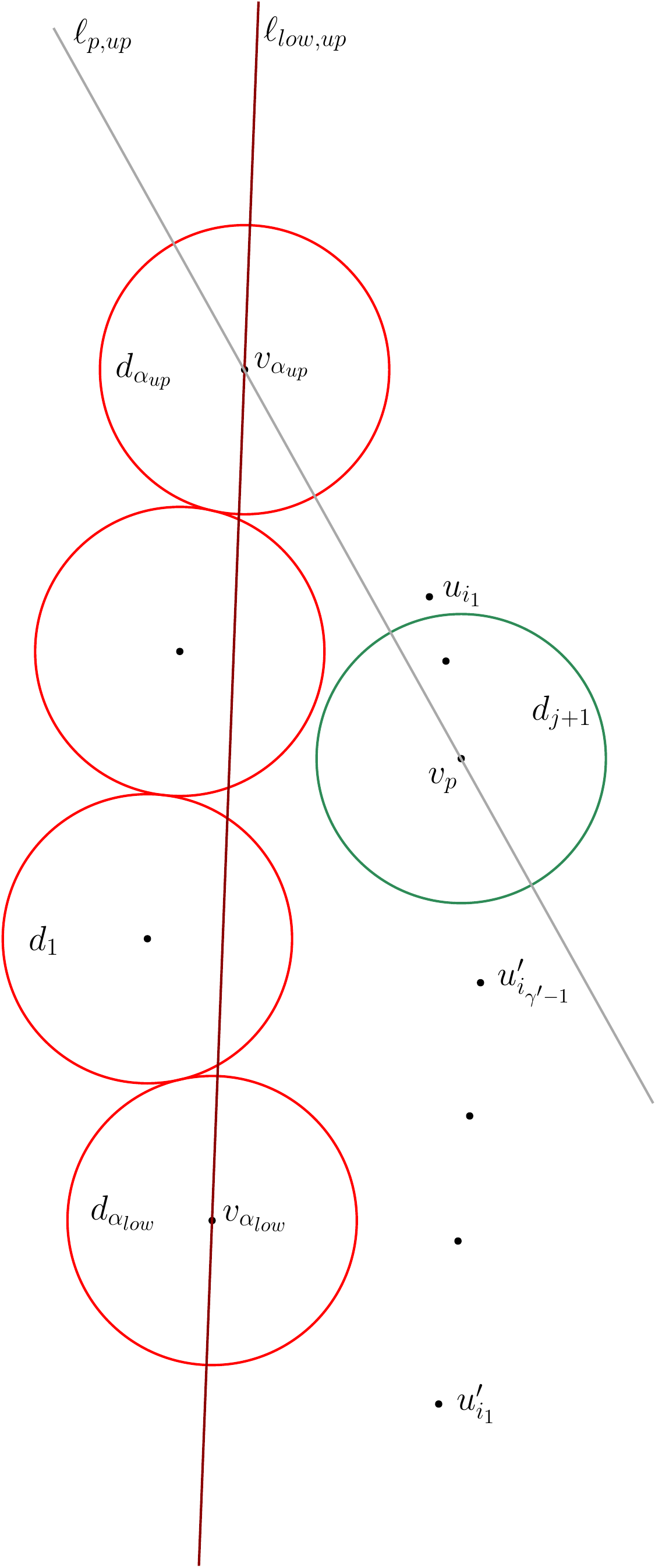}
\caption{Preprocessing and computation of a center vertex for $d_{j+1}$}
\label{duplow3}
\end{figure}

%

\begin{claim} 
If \textsc{Decision}$({\cal P},k,r)=${\textsc{yes}} then there exists a vertex $v_{\alpha_1}$ of ${\cal P}$ with $v_{\alpha_1}={\cal C}(d_1)$ which results in the following: there is a root-leaf path of a 2-way search tree with the root corresponding to $v_{\alpha_1}$ such that at every node along the path we are able to pack the next disk $d_{j+1}$ centered at one of the candidate vertices.
\end{claim}
\begin{proof}
 The proof follows by the above discussion due to the convexity of ${\cal P}$ (See Figure \ref{duplow3}). 
\end{proof}

 \begin{algorithm}
\caption{Exact-fixed-parameter}\label{algorithm1}
\textbf{Input:} A convex polygon ${\cal P}$ with $V$ vertices and an integer $k$\\
\textbf{Output:} Radius $r_{max}$ of $k$ disks packed\\

 $\mathcal{X}_1\leftarrow C'$, $i\leftarrow 1$\\
\While{$|\mathcal{X}_i|\geq 2$}{
 $r\leftarrow$ \textit{median}($\mathcal{X}_i$)/2 \hspace{2mm} /* invoke the linear time median finding algorithm \cite{CORM2009} */\\
Based on the value of $2r$, precompute the candidate center vertices for each $v_s\in V$ and store in a global array $A_s$, $s=1,2,\ldots, n$. \\
\If{\textsc{Decision}$({\cal P},k,r)$}{
$\mathcal{X}_{i+1}\leftarrow \mathcal{X}_i\setminus\{e\in \mathcal{X}_i | e\leq 2r\}$}
\Else{
$\mathcal{X}_{i+1}\leftarrow \mathcal{X}_i\setminus\{e\in \mathcal{X}_i | e\geq 2r\}$
}
$i\leftarrow i+1$\\
}
$r=\min(\mathcal{X}_{i})/2$ \\
\Return {$r$}
\end{algorithm}

  \begin{lemma}\label{label-4}
   We can answer the decision question \textsc{Decision}$({\cal P},k,r)$ in $O(2^k(n^2+n(\log n)(\log k)))$ time.
  \end{lemma}
 \begin{proof}
  The correctness of our decision algorithm follows due to the following facts: \\
  
  \noindent (i) Fix some vertex $v$ of ${\cal P}$ and a center of the disk $d_1$ at $v$. In the search space corresponding to the vertex $v$, along any (root-leaf) path (of the search tree) after the disk $d_j$ is centered, by the claim above if $r\leq r_{max}$ there is always a candidate center vertice $u$ in at least one direction along the boundary of ${\cal P}$ in order to pack the next disk $d_{j+1}$, for $j=2, 3, \ldots, k-1$. We argued that the amortized time for finding this candidate vertex is $O(n+(\log n)(\log{j}))$ (by accessing the array $A$ computed in step 4 of Algorithm \ref{algorithm1}). Hence, the branching factor of every node of the search space is at most 2. Therefore, after the disk $d_1$ is centered at some vertex $v$ of ${\cal P}$, the resulting search space for the remaining $k-1$ disks is a 2-way search tree, and its depth is $O(k)$. 
  
   \noindent (ii) Consider an element $2r\in C'$ such that $r\geq r_{max}$. Now, for this radius $r$ let $k'$ be the maximum number of disks that can be packed in the optimal packing $\textsc{OPT}$ and $k\geq k'$. We can determine a vertex which is the center for the disk $d_1^{\textsc {OPT}}$ in $\textsc{OPT}$, in $O(n)$ time by exploring the search space rooted corresponding to each vertex of the polygon ${\cal P}$. Then, by walking along $\partial {\cal P}$ from the point ${\cal C}(d_1^{\textsc {OPT}})$, we can charge each disk $d_{j'}^{\textsc {OPT}}$ with at least one disk $d_{j}$ centered by Algorithm \ref{algorithm1} if $(d_{j} \cap d_{j'}^{\textsc {OPT}})\neq \emptyset$ for $j'=1,2, \ldots, k'$. Therefore, ${\cal C}(d_1^{\textsc {OPT}})={\cal C}(d_1)$, i.e., $d_1^{\textsc {OPT}}$ gets charged with $d_1$ itself. Suppose there is some optimal disk $d_{j'}^{\textsc {OPT}}$ that does not get charged with any disk $d_j$ centered by Algorithm \ref{algorithm1}. Then this will contradict with the termination of Algorithm \ref{algorithm1} as $k'\geq (k+1)$.
  
   \noindent (iii) In the 2-way search tree there are at most $2^j$ nodes at level $j$. Since we invest $O(n+(\log n)(\log{j}))$ time at every node of the level $j$, the total time will be
   $\sum\limits_{j=1}^{k-1}(2^j(n+(\log n)(\log{j})))=O(2^k(n+(\log n)(\log k)))$.
  If the radius $r\leq r_{max}$, then we can answer \textsc{Decision}$({\cal P},k,r)$ correctly in time $n\cdot(2^k(n+(\log n)(\log k)))=O(2^k(n^2+n(\log n)(\log k)))$ since we exhaustively search all $n$ 2-way search trees. 
 \end{proof}

  \begin{theorem}
   We have an exact fixed-parameter algorithm for the {\textsc{DkConP}} problem in $O(2^k(n^2\log n+n(\log n)^2(\log k)))$ time.
  \end{theorem}
\begin{proof}
 Follows from Lemma \ref{label-4} and by doing binary search on the set $C'$, since the total size of the search tree is bounded by a function of the parameter $k$ alone, and every step takes polynomial time, and there are at most $\lceil 2\log n\rceil$ calls to \textsc{Decision}$({\cal P},k,r)$.
\end{proof}

\section{A sublinear time $\sfrac{1}{2\sqrt{2}}$-approximation for $k=3$}
In this section, we propose an approximation algorithm for the {\textsc{DkConP}} problem in $O(\log{n})$ time for $k=3$ if the points are given in convex position order in the first quadrant. Consider a set $S$ of $n$ points in general convex position. Let $a,b,c,d$ be the four extreme points of the convex polygon ${\cal P}$ formed by the given $n$ points. Without loss of generality let $a$ be the extreme point with the minimum $x$-coordinate, $b$ with the maximum $y$-coordinate, $c$ with the maximum $x$-coordinate, and $d$ with the minimum $y$-coordinates.
 Based on the position of the given $n$ points, we have three possible situations for the positions of extreme points of ${\cal P}$: (i) all the four points $a,b,c,d$ are distinct, (ii) two among four points coincide (i.e., $a=b,c,d$), and (iii) two pair of points coincide (i.e., $a=b,c=d$).

%
%

\begin{figure}[!htb]
\centering
\includegraphics[scale=0.45]{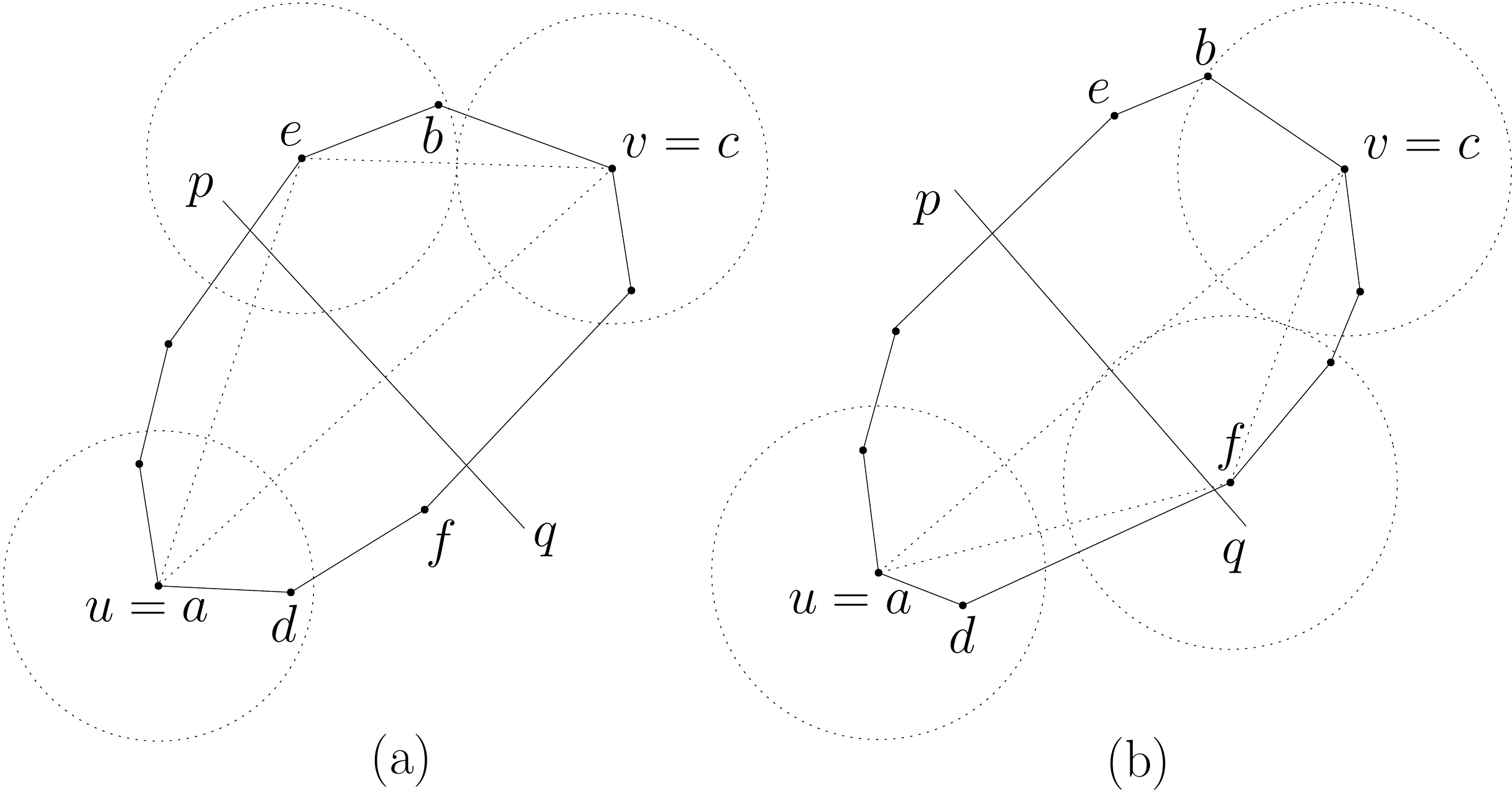}
\caption{(a) $e$ is the farthest point from the line segment $\overline{ac}$ (b) $f$ is the nearest point to the perpendicular bisector $\overline{pq}$}
\label{figure-s3}
\end{figure}


\begin{figure}[!htb]
\centering
\includegraphics[scale=0.45]{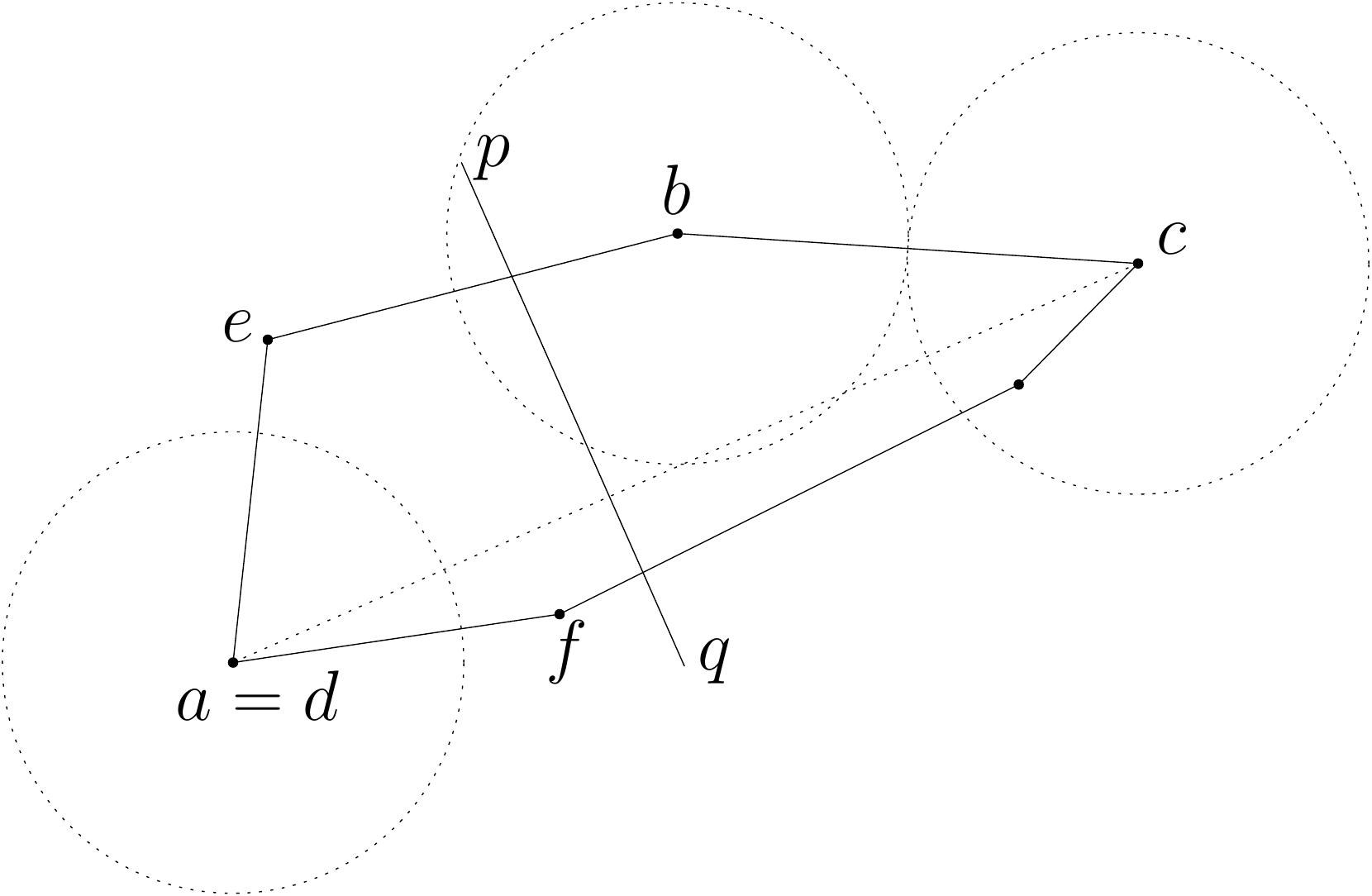}
\caption{The disks are centered at extreme points $a$, $b$ and $c$}
\label{figure-s5}
\end{figure}

Let $\overline{pq}$ be the perpendicular bisector for every line segment joining the extreme points of the polygon, i.e., $\overline{uv}$ for every $u\neq v \in \{a,b,c,d \}$. Now consider the following cases:\\
\textbf{case 1:} Let $e$ be a point of $S$ farthest from the line through $\overline{uv}$. Now, we place disks centered at $u$, $v$ and $e$ such that radius $r_{l1}=\min\{|uv|, |ue|,|ev|\}/{2}$ (see Figure \ref{figure-s3}(a)). This is repeated for every $(u,v)\in \{a,b,c,d \}$, $u\neq v$ and let the corresponding radii be $r_{l1}$ for $l=1,2,\dots,{4\choose 2}$. \\
\textbf{case 2:} Let $f\in S$ be the nearest point to the perpendicular bisector $\overline{pq}$. Now, we place disks centered at $u$, $v$ and $f$ such that radius $r_{l2}=\min\{|uv|, |uf|,|fv|\}/{2}$ (see Figure \ref{figure-s3}(b)). This is repeated for every $(u,v)\in \{a,b,c,d \}$, $u\neq v$ and $r_{l2}$ is the corresponding radii for $l=1,2,\dots,{4 \choose 2}$.\\
\textbf{case 3:} We place disks centered at three of four extreme points such that their common radius is maximized, i.e., \[r_1=\max\limits_{{u\neq v \neq w \in \{a,b,c,d \}}}\{ min \{ |uv|,|vw|,|uw| \} \}/2.\]

Let $\{r_{l1},r_{l2}\}_{l=1}^6=\bigcup\limits_{l=1}^6\{r_{l1},r_{l2}\}$. In the remainder of this section, we argue that the $\max( \{r_1\} \cup \{r_{l1},r_{l2}\}_{l=1}^6)$ will be an $\sfrac{1}{2\sqrt{2}}$-approximate value for the radius in the optimal packing when $k=3$, i.e., $\max( \{r_1\} \cup \{r_{l1},r_{l2}\}_{l=1}^6) \geq r_{max}/2\sqrt{2}$.

%

 We can find the extreme points of the convex polygon in $O(\log{n})$ time \cite{ORO1998}. Let $d_l$ be a diagonal formed by any two of these extreme points. Now, we can observe that the boundary $\partial{\cal P}$ can be split into two convex polygonal chains above and below $d_l$ respectively. Hence, due to the convexity of ${\cal P}$ we can find the farthest point from the line through $d_l$ with a similar adaption of the binary search \cite{ORO1998} in $O(\log {n})$ time. We can find a perpendicular bisector of the line through $d_l$ in constant time and the points of the convex polygon closest to the perpendicular bisector in $O(\log {n})$ time with the similar approach (because here, unimodality property holds for the distances between bisector line and vertices). In total there are at most four such points. Hence, in all the above cases the running time required is $O(\log{n})$, hence we have the following theorem.
 \begin{theorem}
We have a logarithmic time $\sfrac{1}{2\sqrt{2}}$-approximation algorithm for the max-min 3-dispersion problem on a convex polygon.
\end{theorem}
\begin{proof}
The running time of the algorithm is clearly $O(\log {n})$, as every step of the algorithm takes either $O(1)$ or $O(\log{n})$ time. 
Now, we argue about its approximation factor. Let $r_{max}$ be the optimal radius of the disks in the optimal solution. Then $r_{max}$ is at most half of the diameter of ${\cal P}$, i.e., $r_{max}\leq \frac{{\cal D(\cal P)}}{2}$. Let the convex polygon ${\cal P}$ be enclosed inside an axis-parallel rectangle ${\cal R}$ such that the extreme points of ${\cal P}$ lie on the boundary of ${\cal R}$ (see Figure \ref{figure-s9}).
\\
\textbf{case 1:} When $r_1>\max(\{r_{l1},r_{l2}\}_{l=1}^6\setminus \{r_1\})$. Without loss of generality assume that $r_1\geq |ad|/2$ and that $|a'd|/2= \max( \{r_{l1},r_{l2}\}^6_{l=1}\setminus \{r_1\})$ and thus $|a'd|\leq |ad|$ \\
Let $\alpha_1,\alpha_2,\alpha_3, \alpha_4$ be top-left, top-right, bottom-right, bottom-left corner vertices of ${\cal R}$. In figure \ref{figure-s9}(a) we can clearly observe that $|ad|\geq \frac{|\alpha_3 \alpha_4|}{2}$, because otherwise it would imply $|ad|<|a'd|$ contradicting the choice of $r_1$. Let $\theta_1$ be the angle at $\alpha_3$ defined by $\overline{\alpha_1\alpha_3}$ and $\overline{\alpha_4\alpha_3}$. Observe that $\theta_1\leq 45^{\circ}$. Therefore,
\begin{equation}\label{eq1}
cos(\theta_1)\geq \frac{1}{\sqrt{2}} \implies
\frac{|\alpha_4\alpha_3|}{|\alpha_1\alpha_3|}\geq \frac{1}{\sqrt{2}} \implies |\alpha_4\alpha_3|\geq \frac{1}{\sqrt{2}}|\alpha_1\alpha_3|
\end{equation}
Since $r_{max}\leq \frac{|\alpha_1\alpha_3|}{2}$ and $|ad|\geq \frac{|\alpha_3 \alpha_4|}{2}$ where $r_1=\frac{|ad|}{2}$, then we have,
\[r_1\geq \frac{|\alpha_3 \alpha_4|}{4}\geq  \frac{1}{4\sqrt{2}}|\alpha_1\alpha_3|\geq  \frac{1}{2\sqrt{2}}r_{max} \]

\begin{figure}[!htb]
\centering
\includegraphics[scale=0.45]{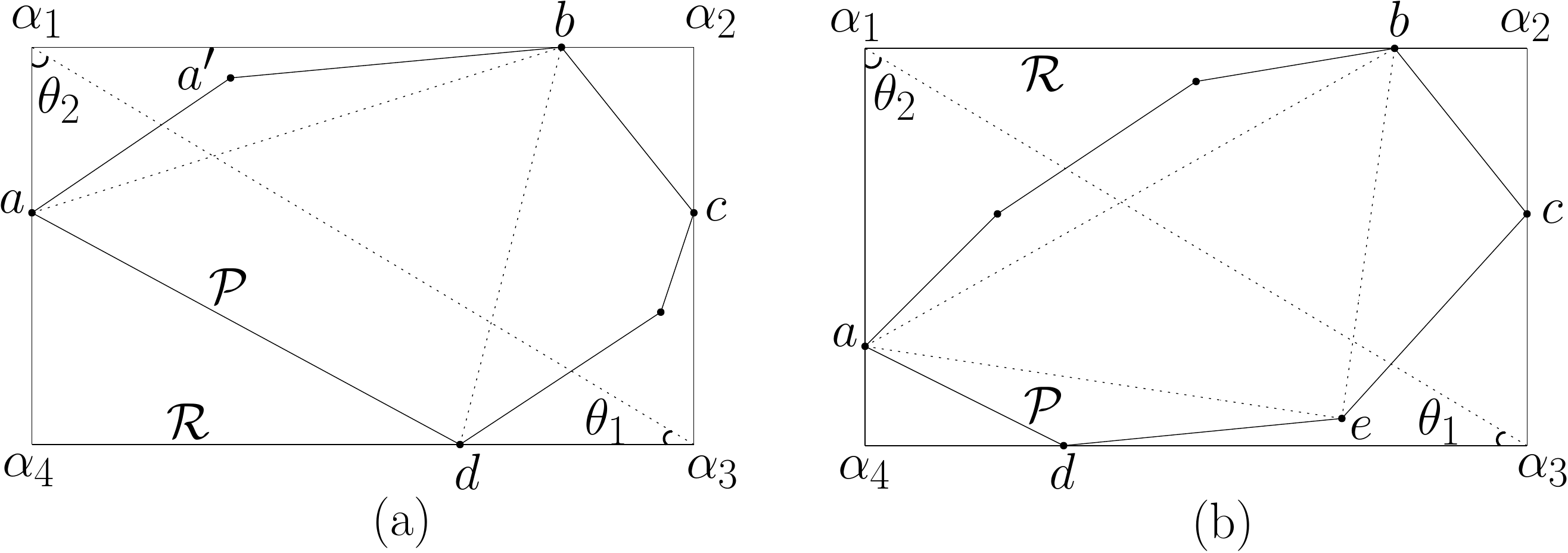}
\caption{(a) When $r_1> \max (\{r_{l1},r_{l2}\}_{l=1}^6\setminus \{r_1\})$    (b) When $r_{lm*}= \max{ \{r_{l1},r_{l2}\}_{l=1}^6}$ and $r_{lm*}>r_1$}
\label{figure-s9}
\end{figure}

\noindent \textbf{case 2:} When $r_{l^*m}=\max{\{r_{l1},r_{l2}\}_{l=1}^6}$ and $r_{l^*m}>r_1$, where $l^*\in \{1,2,\ldots, 6\}$, and $m=1$ or $m=2$. \\
Let $e$ be the farthest point from the line through the segment $\overline{ab}$. Without loss of generality assume that $r_{l^*m}$ corresponds to $\{a, b, e\}$. In figure \ref{figure-s9}(b) we can clearly observe that $|ae|\geq \frac{|\alpha_3 \alpha_4|}{2}$, because otherwise it would imply that any vertex $f$ above $\overline{ab}$ along with $c$ and $e$ corresponds to $r_{l^*m}$ contradicting the choice of $r_{l^*m}$. Therefore, we have the same series of inequalities (\ref{eq1}).
Since $r_{max}\leq \frac{|\alpha_1\alpha_3|}{2}$ and $|ae|\geq \frac{|\alpha_3 \alpha_4|}{2}$ where $r_{l^*m}=\frac{|ae|}{2}$, then we have,
\[r_{l^*m}\geq \frac{|\alpha_3 \alpha_4|}{4}\geq  \frac{1}{4\sqrt{2}}|\alpha_1\alpha_3|\geq  \frac{1}{2\sqrt{2}}r_{max} \]

Hence, the algorithm chooses two disk centers at the endpoints of $d_l$ while the third disk is placed based on the above mentioned three cases, which gives the radius that is at least $ \frac{r_{max}}{2\sqrt{2}}$. 
\end{proof}

\section{Conclusion}\label{section-5}

In this paper, we studied the $k$-dispersion problem on a convex polygon and proposed an exponential time exact algorithm for any $k$ and a logarithmic time approximation algorithm for $k=3$. There are many directions for further research on this problem and the general Euclidean $k$-dispersion problem. Since the {\tt NP-hard}ness of the convex version of the problem is unknown, the problem is open from the point of polynomial exact algorithm or faster exact exponential algorithm. The general Euclidean $k$-dispersion problem is open from the point of polynomial time approximation algorithm with a bettter factor than $\frac{1}{2}$.

\small
\bibliographystyle{abbrv}

\end{document}